\documentclass[envcountsame,envcountsect,pdf]{llncs}

\usepackage[latin1]{inputenc}
\usepackage[T1]{fontenc}
\usepackage[english]{babel}
\usepackage{rotating}
\usepackage{amsmath, amssymb, latexsym, amsfonts}
\usepackage{euscript, stmaryrd, bigstrut}
\usepackage{epsfig,pstricks, pst-node, pst-tree, pst-coil}
\usepackage{listings}
\usepackage{galois}
\usepackage{ifthen}

\newcommand{\setR}{\mathbb{R}}
\newcommand{\setN}{\mathbb{N}}
\newcommand{\setZ}{\mathbb{Z}}
\newcommand{\setQ}{\mathbb{Q}}
\newcommand{\F}{\EuScript{F}}
\newcommand{\Lan}{L}
\newcommand{\finiteset}[2]{\{#1,\ldots,#2\}}
\newcommand{\moins}{\backslash}
\newcommand{\fo}[1]{\operatorname{FO}\left(#1\right)}
\newcommand{\C}{\EuScript{C}}

\newcommand{\clo}[1]{\operatorname{cl}(#1)}
\newcommand{\conv}[1]{\operatorname{conv}(#1)}
\newcommand{\cloconvop}{\operatorname{cl}\circ\operatorname{conv}}
\newcommand{\cloconv}[1]{\cloconvop(#1)}
\newcommand{\partie}[1]{\mathcal{P}(#1)}

\newcommand{\ifthenelsepaperversion}[3]{%
  \ifthenelse{%
    \equal{#1}{\paperversion}}%
  {#2}{#3}%
}

\newcommand{\onlyforpaperversion}[2]{%
  \ifthenelsepaperversion{#1}%
  {#2}{}%
}

\lstdefinelanguage{myalgo}{%
  morekeywords={if,then,else,repeat,while,for,to,forall,do,compute,call,return}
}
\begin{document}
\newcommand{\paperversion}{full}

\ifthenelsepaperversion{proc}{}{\pagestyle{plain}}

\ifthenelsepaperversion{proc}{}{\pagestyle{plain}}

\title{Convex Hull of Arithmetic Automata}

\author{Jérôme Leroux}

\institute{
  LaBRI, Université de Bordeaux, CNRS\\
  Domaine Universitaire, 351, cours de la Libération, 33405 Talence, France\\
  \email{leroux@labri.fr}
}

\maketitle

\onlyforpaperversion{draft}{
\begin{center}
  \psshadowbox[framesep=5pt]{
    \bfseries
    Submitted  --- Please Do Not Circulate --- \textit{\today}
  }
\end{center}
} 

\begin{abstract}
Arithmetic automata recognize infinite words of digits denoting decompositions of real and integer vectors. These automata are known expressive and efficient enough to represent the whole set of solutions of complex linear constraints combining both integral and real variables. In this paper, the closed convex hull of arithmetic automata is proved rational polyhedral. Moreover an algorithm computing the linear constraints defining these convex set is provided. Such an algorithm is useful for effectively extracting geometrical properties of the whole set of solutions of complex constraints symbolically represented by arithmetic automata.
\end{abstract}



\section{Introduction}

The \emph{most significant digit first decomposition} provides a natural way to associate finite words of digits to any integer. Naturally, such a decomposition can be extended to real values just by considering \emph{infinite} words rather than \emph{finite} ones. Intuitively, an infinite word denotes the potentially infinite decimal part of a real number. Last but not least, the most significant digit first decomposition can be extended to real vectors just by interleaving the decomposition of each component into a single infinite word.

\medskip

\emph{Arithmetic automata} are Muller automata that recognize infinite words of most significant digit first decompositions of real vectors in a fixed basis of decomposition $r\geq 2$ (for instance $r=2$ and $r=10$ are two classical basis of decomposition). Sets symbolically representable by arithmetic automata in basis $r$ are logically characterized \cite{DBLP:conf/icalp/BoigelotRW98} as the sets definable in the first order theory $\fo{\setR,\setZ,+,\leq,X_r}$ where $X_r$ is an additional predicate depending on the basis of decomposition $r$. 
In practice, arithmetic automata are usually used for the first order additive theory $\fo{\setR,\setZ,+,\leq}$ where $X_r$ is discarded. In fact this theory allows to express complex linear constraints combining both integral and real variables that can be represented by particular Muller automata called \emph{deterministic weak Buchi automata} \cite{DBLP:journals/tocl/BoigelotJW05}.
This subclass of Muller automata has interesting algorithmic properties. In fact, compared to the general class, deterministic weak Buchi automata can be minimized (for the number of states) into a unique canonical form with roughly the same algorithm used for automata recognizing finite words. In particular, these arithmetic automata are well adapted to symbolically represent sets definable in $\fo{\setR,\setZ,+,\leq}$ obtained after many operations (boolean combinations, quantifications). In fact, since the obtained arithmetic automata only depends on the represented set and not on the potentially long sequence of operations used to compute this set, we avoid unduly complicated arithmetic automata. Intuitively, the automaton minimization algorithm performs like a simplification procedure for $\fo{\setR,\setZ,+,\leq}$. In particular arithmetic automata are adapted to the symbolic model checking approach computing inductively reachability sets of systems manipulating counters \cite{DBLP:conf/cav/BardinLP06} and/or clocks \cite{DBLP:conf/cav/BoigelotH06}. In practice algorithms for effectively computing an arithmetic automaton encoding the solutions of formulas in $\fo{\setR,\setZ,+,\leq}$ have been recently successfully implemented in tools \textsc{Lash} and \textsc{Lira} \cite{DBLP:conf/cav/BeckerDEK07}.
Unfortunately, interesting qualitative properties are difficult to extract from arithmetic automata. Actually, operations that can be performed on the arithmetic automata computed by tools \textsc{Lash} and \textsc{Lira} are limited to the universality and the emptiness checking (when the set symbolically represented is not empty these tools can also compute a real vector in this set).

\medskip

Extracting geometrical properties from an arithmetic automaton representing a set $X\subseteq \setR^m$ is a complex problem even if $X$ is definable in $\fo{\setR,\setZ,+,\leq}$. Let us recall related works to this problem. Using a Karr based algorithm \cite{DBLP:journals/acta/Karr76}, the affine hull of $X$ has been proved efficiently computable in polynomial time \cite{DBLP:journals/entcs/Leroux04} (even if this result is limited to the special case $X\subseteq \setN^m$, it can be easily extended to any arithmetic automata). When $X=\setZ^m\cap C$ where $C$ is a rational polyhedral convex set (intuitively when $X$ is equal to the integral solutions of linear constraint systems), it has been proved in \cite{DBLP:conf/lics/Latour04} that we can effectively compute in exponential time a rational polyhedral convex set $C'$ such that $X=\setZ^m\cap C'$. Note that this worst case complexity in theory is not a real problem in practice since the algorithm presented in \cite{DBLP:conf/lics/Latour04} performs well on automata with more than 100~000 states. In \cite{DBLP:conf/wia/Lugiez04} this result was extended to sets $X=F+L$ where $F$ is a finite set of integral vectors and $L$ is a linear set. In \cite{DBLP:journals/ipl/FinkelL05}, closed convex hulls of sets $X\subseteq \setZ^m$ represented by arithmetic automata are proved rational polyhedral and effectively computable in exponential time. Note that compared to \cite{DBLP:conf/lics/Latour04}, it is not clear that this result can be turn into an efficient algorithm. More recently \cite{DBLP:conf/lics/Leroux05}, we provided an algorithm for effectively computing in polynomial time a formula in the Presburger theory $\fo{\setZ,+,\leq}$ when $X\subseteq \setZ^n$ is Presburger-definable. This algorithm has been successfully implemented in \textsc{TaPAS} \cite{LP-08} (The Talence Presburger Arithmetic Suite) and it can be applied on any arithmetic automata encoding a set $X\subseteq\setZ^m$ with more than 100~000 states. Actually, the tool decides if an input arithmetic automaton denotes a Presburger-definable set and in this case it returns a formula denoting this set.

\medskip

In this paper we prove that the closed convex hulls of sets symbolically represented by arithmetic automata are rational polyhedral and effectively computable in exponential time in the worst case. Note that whereas the closed convex hull of a set definable in $\fo{\setR,\setZ,+,\leq}$ can be easily proved rational polyhedral (thanks to quantification eliminations), it is difficult to prove that the closed convex hulls of arithmetic automata are rational polyhedral. We also provide an algorithm for computing this set. Our algorithm is based on the reduction of the closed convex hull computation to data-flow analysis problems. Note that widening operator is usually used in order to speed up the iterative computation of solutions of such a problem. However, the use of widening operators may lead to loss of precision in the analysis. Our algorithm is based on \emph{acceleration} in convex data-flow analysis \cite{DBLP:conf/fsttcs/LerouxS07,DBLP:conf/sas/LerouxS07}. Recall that acceleration consists to compute the exact effect of some control-flow cycles in order to speed up the Kleene fix-point iteration.
 
\medskip

\emph{Outline of the paper} : In section \ref{sec:arithmetic} the most significant digit first decomposition is extended to any real vector and we introduce the arithmetic automata. In section \ref{sec:reduction} we provide the closed convex hull computation reduction to (1) a data-flow analysis problem and (2) the computation of the closed convex hull of arithmetic automata representing only decimal values and having a trivial accepting condition. In section \ref{sec:infinite} we provide an algorithm for computing the closed convex hull of such an arithmetic automaton. Finally in section \ref{sec:computation} we prove that the data-flow analysis problem introduced by the reduction can be solved precisely with an accelerated Kleene fix-point iteration algorithm.
\ifthenelsepaperversion{proc}{
Due to space limitations, most proofs are only sketched in this paper.  A long version of the paper with detailed proofs can be obtained from the author.
}%
{%
Most proofs are only sketched in the paper, but detailed proofs are given in appendix.  This paper is the long version of the SAS 2008 paper.
}

\section{Arithmetic Automata}\label{sec:arithmetic}
This section introduces arithmetic automata (see Fig. \ref{fig:ex}). These automata recognize infinite words of digits denoting \emph{most significant digit first} decompositions of real and integer vectors.

\medskip

As usual, we respectively denote by $\setZ$, $\setQ$ and $\setR$ the sets of integers, rationals and real numbers and we denote by $\setN, \setQ_+, \setR_+$ the restrictions of $\setZ, \setQ, \setR$ to the non-negatives. The \emph{components} of an $m$-dim vector $x$ are denoted by $x[1],\ldots,x[m]$. 

\medskip

We first provide some definitions about regular sets of infinite words. We denote by $\Sigma$ a non-empty finite set called an \emph{alphabet}. 
An \emph{infinite word} $w$ over $\Sigma$ is a function $w\in\setN\rightarrow\Sigma$ defined over $\setN\moins\{0\}$ and a \emph{finite word} $\sigma$ over $\Sigma$ is a function $\sigma\in\setN\rightarrow\Sigma$ defined over a set $\{1,\ldots,k\}$ where $k\in\setN$ is called the \emph{length} of $\sigma$ and denoted by $|\sigma|$. \emph{In this paper, a finite word over $\Sigma$ is denoted by $\sigma$ with some subscript indices and an infinite word over $\Sigma$ is denoted by $w$}. As usual $\Sigma^*$ and $\Sigma^\omega$ respectively denote the set of finite words and the set of infinite words over $\Sigma$. The concatenation of two finite words $\sigma_1,\sigma_2\in\Sigma^*$ and the concatenation of a finite word $\sigma\in\Sigma^*$ with an infinite word $w\in\Sigma^\omega$ are denoted by $\sigma_1\sigma_2$ and $\sigma w$. 
A \emph{graph labelled} by $\Sigma$ is a tuple $G=(Q,\Sigma,T)$ where $Q$ is a non empty finite set of \emph{states} and $T\subseteq Q\times\Sigma\times Q$ is a set of \emph{transitions}. A \emph{finite path} $\pi$ in a graph $G$ is a finite word $\pi=t_1\ldots t_k$ of $k\geq 0$ transitions $t_i\in T$ such that there exists a sequence $q_0,\ldots,q_k\in Q$ and a sequence $a_1,\ldots,a_k\in \Sigma$ such that $t_i=(q_{i-1},a_i,q_i)$ for any $1\leq i\leq k$. The finite word $\sigma=a_1\ldots a_k$ is called the \emph{label} of $\pi$ and such a path $\pi$ is also denoted by $q_0\xrightarrow{\sigma}q_k$ or just $q_0\rightarrow q_k$. We also say that $\pi$ is a path \emph{starting} from $q_0$ and \emph{terminating} in $q_k$. When $q_0=q_k$ and $k\geq 1$, the path $\pi$ is called a cycle on $q_0$. Such a cycle is said \emph{simple} if the states $q_0,\ldots,q_{k-1}$ are distinct. Given an integer $m\geq 1$, a graph $G$ is called an $m$-graph if $m$ divides the length of any cycle in $G$. An \emph{infinite path} $\theta$ is an infinite word of transitions such that any prefixes $\pi_k=\theta(1)\ldots \theta(k)$ is a finite path. The unique infinite word $w\in\Sigma^\omega$ such that $\sigma_k=w(1)\ldots w(k)$ is the label of the finite path $\pi_k$ for any $k\in\setN$ is called the \emph{label} of $\theta$. We say that $\theta$ is \emph{starting} from $q_0$ if $q_0$ is the unique state such that any prefix of $\theta$ is starting from $q_0$. \emph{In the sequel, a finite path is denoted by $\pi$ and an infinite path is denoted by $\theta$}. The \emph{set of infinite paths starting from $q_0$} is naturally denoted with the capital letter $\Theta_G(q_0)$. The set $F$ of states $q\in Q$ such that there exists an infinite number of prefix of $\theta$ terminating in $q$ is called the set of \emph{states visited infinitely often} by $\theta$. Such a path is denoted by $q_0\xrightarrow{w}F$ or just $q_0\rightarrow F$.
A \emph{Muller automaton} $A$ is a tuple $A=(Q,\Sigma,T,Q_0,\F)$ where $(Q,\Sigma,T)$ is a graph, $Q_0\subseteq Q$ is the \emph{initial condition} and $\F\subseteq \partie{Q}$ is the \emph{accepting condition}. The language $\Lan(A)\subseteq\Sigma^\omega$ \emph{recognized} by a Muller automaton $A$ is the set of infinite words $w\in\Sigma^\omega$ such that there exists an infinite path $q_0\xrightarrow{w}F$ with $q_0\in Q_0$ and $F\in\F$.

\begin{figure}[htbp]
  \begin{center}
    \small
    \setlength{\unitlength}{10pt}
    \pssetlength{\psunit}{10pt}
    \pssetlength{\psxunit}{10pt}
    \pssetlength{\psyunit}{10pt}
    \begin{picture}(27,17)(-14,-9)%
      \rput(0,0){\includegraphics{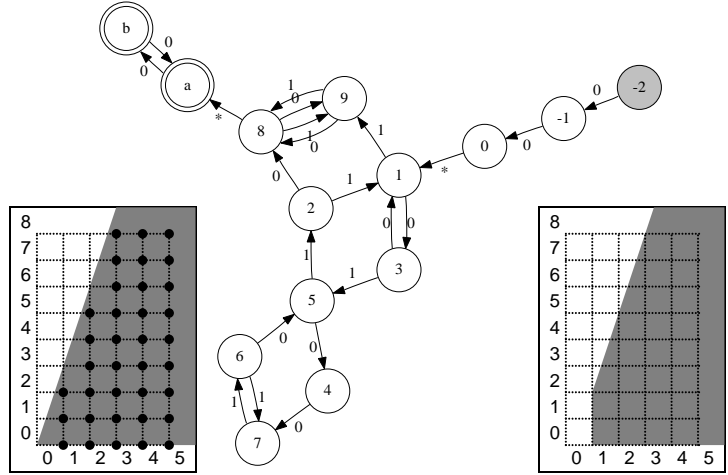}}
      \rput(-13,-4){%
      \put(-1,-1){\framebox(7,10){}}
      \pspolygon*[linecolor=gray](0,0)(3,9)(6,9)(6,0)
      %
      %
      \psgrid[subgriddiv=1,gridcolor=black,griddots=8,gridlabels=7pt](0,0)(5,8)
      \psdots(1,0)(1,1)(1,2)
      \psdots(2,0)(2,1)(2,2)(2,3)(2,4)(2,5)
      \psdots(3,0)(3,1)(3,2)(3,3)(3,4)(3,5)(3,6)(3,7)(3,8)
      \psdots(4,0)(4,1)(4,2)(4,3)(4,4)(4,5)(4,6)(4,7)(4,8)
      \psdots(5,0)(5,1)(5,2)(5,3)(5,4)(5,5)(5,6)(5,7)(5,8)
      %
    }
      \rput(7,-4){%
	\put(-1,-1){\framebox(7,10){}}
	\pspolygon*[linecolor=gray](1,0)(1,2)(3.333,9)(6,9)(6,0)
	\psgrid[subgriddiv=1,gridcolor=black,griddots=8,gridlabels=7pt](0,0)(5,8)
	%
	%
	%
      }
    \end{picture}
  \end{center}
  \caption{%
    On the left, the rational polyhedral convex set $C=\{x\in\setR^2 \mid 3x[1]>x[2] \wedge x[2]\geq 0\}$ in gray and the set $X=\setZ^2\cap C$ of integers depicted by black bullets. On the center, an arithmetic automaton symbolically representing $X$ in basis $2$. On the right, the closed convex hull of $X$ equals to $\cloconv{X}=\{x\in\setR^2 \mid 3x[1]\geq x[2]+1 \wedge x[2]\geq 0\wedge x[1]\geq 1\}$ represented in gray.}
  \label{fig:ex}
\end{figure}


\medskip

Now, we introduce the \emph{most significant digit first decomposition} of real vectors. In the sequel $m\geq 1$ is an integer called \emph{the dimension}, $r\geq 2$ is an integer called the \emph{basis of decomposition}, $\Sigma_r=\finiteset{0}{r-1}$ is called the alphabet of \emph{$r$-digits}, and $S_r=\{0,r-1\}$ is called the alphabet of \emph{sign $r$-digits}. The most significant $r$-digit first decomposition provides a natural way to associate to any real vector $x\in\setR^m$ a tuple $(s,\sigma,w)\in S_r^m\times(\Sigma_r^m)^*\times\Sigma_r^\omega$. Intuitively $(s,\sigma)$ and $w$ are respectively associated to an integer vector $z\in\setZ^m$ and a decimal vector $d\in\left[0,1\right]^m$ satisfying $x=z+d$. Moreover, $s[i]=0$ corresponds to $z[i]\geq 0$ and $s[i]=r-1$ corresponds to $z[i]<0$. More formally, a \emph{most significant $r$-digit first decomposition} of a real vector $x\in\setR^m$ is a tuple $(s,\sigma,w)\in S_r^m\times(\Sigma_r^m)^*\times \Sigma_r^\omega$ such that for any $1\leq i\leq m$, we have:
$$
x[i]
=
r^{\frac{|\sigma|}{m}}\frac{s(i)}{1-r}
+
\sum_{j=1}^{\frac{|\sigma|}{m}} r^{\frac{|\sigma|}{m}-j} \sigma(m(j-1)+i)+\sum_{j=0}^{+\infty}\frac{w(mj+i)}{r^{j+1}}
$$
The previous equality is divided in two parts by introducing the functions $\lambda_{r,m}\in\Sigma_r^\omega\rightarrow\left[-1,0\right]^m$ and $\gamma_{r,m}\in S_r^m\times (\Sigma_r^m)^*\rightarrow\setZ^m$ defined for any $1\leq i\leq m$ by the following equalities. Note the sign in front of the definition of $\lambda_{r,m}$. This sign simplifies the presentation of this paper and it is motivated in the sequel.
$$
\begin{array}{l}
\displaystyle

-\lambda_{r,m}(w)[i]
=
\sum_{j=0}^{+\infty}\frac{w(mj+i)}{r^{j+1}}
\\
\displaystyle
\gamma_{r,m}(s,\sigma)[i]
=
r^{\frac{|\sigma|}{m}}\frac{s(i)}{1-r}
+
\sum_{j=1}^{\frac{|\sigma|}{m}} r^{\frac{|\sigma|}{m}-j} \sigma(m(j-1)+i)
\end{array}
$$

\begin{definition}[\cite{DBLP:conf/icalp/BoigelotRW98}]\label{def:ar}
  An \emph{arithmetic automaton} $A$ in basis $r$ and in dimension $m$ is a Muller automaton over the alphabet $\Sigma_r\cup\{\star\}$ that recognizes a language $\Lan\subseteq S_r^m\star(\Sigma_r^m)^*\star\Sigma_r^\omega$. The following set $X\subseteq \setR^m$ is called the set symbolically represented by $A$:
  $$X=\{\gamma_{r,m}(s,\sigma)-\lambda_{r,m}(w) \mid s\star\sigma\star w\in\Lan\}$$
\end{definition}

\begin{example}
  The arithmetic automaton depicted in Fig.~\ref{fig:ex} symbolically represents $X=\{x\in\setN^2 \mid 3x[1]>x[2]\}$. This automaton has been obtained automatically from the tool \textsc{Lash} through the tool-suite \textsc{TaPAS}\cite{LP-08}.
\end{example}

We observe that \emph{Real Vector Automata (RVA)} and \emph{Number Decision Diagrams (NDD)} \cite{DBLP:conf/icalp/BoigelotRW98} are particular classes of arithmetic automata. In fact, RVA and NDD are arithmetic automata $A$ that symbolically represent sets $X$ included respectively in $\setR^m$ and $\setZ^m$ and such that the accepted languages $\Lan(A)$ satisfy:
\begin{align*}
  \Lan(A)=&\{s\star\sigma\star w \mid \gamma_{r,m}(s,\sigma)-\lambda_{r,m}(w)\in X\} && \text{if $A$ is a RVA}\\
  \Lan(A)=&\{s\star\sigma\star 0^\omega \mid \gamma_{r,m}(s,\sigma)\in X\} && \text{if $A$ is a NDD}
\end{align*}
Since in general a NDD is not a RVA and conversely a RVA is not a NDD,  we consider arithmetic automata in order to solve the closed convex hull computation uniformly for these two classes. Note that simple (even if computationally expensive) automata transformations show that sets symbolically representable by arithmetic automata in basis $r$ are exactly the sets symbolically representable by RVA in basis $r$. In particular \cite{DBLP:conf/icalp/BoigelotRW98}, sets symbolically representable by arithmetic automata in basis $r$ are exactly the sets definable in $\fo{\setR,\setZ,+,\leq,X_r}$ where $X_r\subseteq\setR^3$ is a basis dependant predicate defined in \cite{DBLP:conf/icalp/BoigelotRW98}. This characterization shows that arithmetic automata can symbolically represent sets of solutions of complex linear constraints combining both integral and real values. Recall that the construction of arithmetic automata from formulae in $\fo{\setR,\setZ,+,\leq,X_r}$ is effective and tools \textsc{Lash} and \textsc{Lira} \cite{DBLP:conf/cav/BeckerDEK07} implement efficient algorithms for the restricted logic $\fo{\setR,\setZ,+,\leq}$. The predicate $X_r$ is discarded in these tools in order to obtain arithmetic automata that are deterministic weak Buchi automata \cite{DBLP:journals/tocl/BoigelotJW05}. In fact these automata have interesting algorithmic properties (minimization and deterministic form).

\section{Reduction to Data-Flow Analysis Problems}\label{sec:reduction}
In this section we reduce the computation of the closed convex hull of sets symbolically represented by arithmetic automata to data-flow analysis problems.

\medskip 

We first recall some general notions about complete lattices. Recall that a \emph{complete lattice} is any partially ordered set $(A, \sqsubseteq)$ such that every subset $X \subseteq A$ has a \emph{least upper bound} $\bigsqcup X$ and a \emph{greatest lower bound} $\bigsqcap X$.  The \emph{supremum} $\bigsqcup A$ and the \emph{infimum} $\bigsqcap A$ are respectively denoted by $\top$ and $\bot$.  A function $f \in A \rightarrow A$ is \emph{monotonic}
if $f(x) \sqsubseteq f(y)$ for all $x \sqsubseteq y$ in $A$.
%
For any complete lattice $(A, \sqsubseteq)$ and any set $Q$, we also denote by $\sqsubseteq$ the partial order on $Q \rightarrow A$ defined as the point-wise extension of $\sqsubseteq$, i.e. $f \sqsubseteq g$ iff $f(q) \sqsubseteq g(q)$ for all $q \in Q$.  The partially ordered set $(Q \rightarrow A, \sqsubseteq)$ is also a complete lattice, with lub $\bigsqcup$ and glb $\bigsqcap$ satisfying $(\bigsqcup F)(s) = \bigsqcup \, \{f(s) \mid f \in F\}$ and $(\bigsqcap F)(s) = \bigsqcap \, \{f(s) \mid f \in F\}$ for any subset $F \subseteq Q \rightarrow A$.

\medskip

Now, we recall notions about the complete lattice of closed convex sets.
A function $f\in \setR^n\rightarrow\setR^m$ is said \emph{linear} if there exists a sequence $(M_{i,j})_{i,j}$ of reals indexed by $1\leq i\leq m$ and $1\leq j\leq n$ and a sequence $(v_i)_{i}$ of reals indexed by $1\leq i\leq m$ such that $f(x)[i]=\sum_{j=1}^n M_{i,j}x[j]+v_i$ for any $x\in\setR^n$ and for any $1\leq i\leq m$. When the coefficients $(M_{i,j})_{i,j}$ and $(v_i)_{i}$ are rational, the linear function $f$ is said \emph{rational}. The function $f'\in \setR^m\rightarrow\setR^n$ defined by $f'(x)[i]=\sum_{j=1}^n M_{i,j}x[j]$ for any $x\in\setR^n$ and for any $1\leq i\leq m$ is called the \emph{uniform form} of $f$.
A set $R\subseteq \setR^m$ is said \emph{closed} if the limit of any convergent sequence of vectors in $R$ is in $R$. Recall that any set $X\subseteq\setR^m$ is included in a minimal for the inclusion closed set. This closed set is called the \emph{topological closure} of $X$ and it is denoted by $\clo{X}$. 
Let us recall some notions about convex sets (for more details, see \cite{S-87}). A \emph{convex combination} of $k\geq 1$ vectors $x_1,\ldots,x_k\in\setR^m$ is a vector $x$ such that there exists $r_1,\ldots,r_k\in\setR_+$ satisfying $r_1+\cdots+r_k=1$ and $x=r_1x_1+\cdots+r_kx_k$. A set $C\subseteq \setR^m$ is said \emph{convex} if any convex combination of vectors in $C$ is in $C$. Recall that any $X\subseteq \setR^m$ is included in a minimal for the inclusion convex set. This convex set is called the \emph{convex hull} of $X$ and it is denoted by $\conv{X}$. A convex set $C\subseteq \setR^m$ is said \emph{rational polyhedral} if there exists a rational linear function $f\in\setR^m\rightarrow\setR^n$ such that $C$ is the set of vectors $x\in\setR^m$ such that $\bigwedge_{i=1}^nf(x)[i]\leq 0$. 
Recall that $\clo{\conv{X}}=\conv{\clo{X}}$, $\clo{f(X)}=f(\clo{X})$ and $\conv{f(X)}=f(\conv{X})$ for any $X\subseteq\setR^m$ and for any linear function $f\in\setR^m\rightarrow\setR^n$.
The class of closed convex subsets of $\setR^m$ is written $\C_m$.  We denote by $\sqsubseteq$ the inclusion partial order on $\C_m$.  Observe that $(\C_m, \sqsubseteq)$ is a complete lattice, with lub $\bigsqcup$ and glb $\bigsqcap$ satisfying $\bigsqcup \C = \cloconv{\bigcup \C}$ and $\bigsqcap \C = \bigcap \C$ for any subset $\C \subseteq \C_m$.

\begin{example}\label{ex:fin1}
  Let $X=\setZ^2\cap C$ where $C$ is the convex set $C=\{x\in\setR^2 \mid 3x[1]>x[2]\wedge x[2]\geq 0\}$ (see Fig. \ref{fig:ex}). Observe that $\cloconv{X}=\{x\in\setR^2 \mid 3x[1]\geq x[2]+1\wedge x[2]\geq 0\wedge x[1]\geq 1\}$ is strictly included in $C$. 
\end{example}

\medskip

In the previous section, we introduced two functions $\lambda_{r,m}$ and $\gamma_{r,m}$. Intuitively these functions ``compute'' respectively decimal vectors associated to infinite words and integer vectors associated to finite words equipped with sign vectors. We now introduce two functions $\Lambda_{r,m,\sigma}$ and $\Gamma_{r,m,\sigma}$ that ``\emph{partially compute}'' the same vectors than $\lambda_{r,m}$ and $\gamma_{r,m}$. More formally, let us consider the unique sequences $(\Lambda_{r,m,\sigma})_{\sigma\in\Sigma_r^*}$ and $(\Gamma_{r,m,\sigma})_{\sigma\in\Sigma_r^*}$ of linear functions $\Lambda_{r,m,\sigma},\Gamma_{r,m,\sigma}\in\setR^m\rightarrow\setR^m$ inverse of each other and satisfying $\Lambda_{r,m,\sigma_1\sigma_2}=\Lambda_{r,m,\sigma_1}\circ\Lambda_{r,m,\sigma_2}$, $\Gamma_{r,m,\sigma_1\sigma_2}=\Gamma_{r,m,\sigma_2}\circ\Gamma_{r,m,\sigma_1}$ for any $\sigma_1,\sigma_2\in\Sigma_r^*$, such that $\Lambda_{r,m,\epsilon}$ and $\Gamma_{r,m,\epsilon}$ are the identity function and such that $\Lambda_{r,m,a}$ and $\Gamma_{r,m,a}$ with $a\in\Sigma_r$ satisfy the following equalities where $x\in\setR^m$:
\begin{align*}
  \Lambda_{r,m,a}(x)&=(\frac{x[m]-a}{r},x[1],\ldots,x[m-1])\\
  \Gamma_{r,m,a}(x)&=(x[2],\ldots,x[m],rx[1]+a)
\end{align*}
We first prove the following two equalities (\ref{equ:rho}) and (\ref{equ:eta}) that explain the link between the notations $\lambda_{r,m}$ and $\gamma_{r,m}$ and their capital forms $\Lambda_{r,m,\sigma}$ and $\Gamma_{r,m,\sigma}$. Observe that $\Lambda_{r,m,a}(\lambda_{r,m}(w))=\lambda_{r,m}(a w)$ for any $a\in\Sigma_r$ and for any $w\in\Sigma_r^\omega$. An immediate induction over the length of $\sigma\in\Sigma_r^*$ provides equality~(\ref{equ:rho}). Note also that $\Gamma_{r,m,a_1\ldots a_m}(x)=rx+(a_1,\ldots,a_m)$ for any $a_1,\ldots,a_m\in\Sigma_r$. Thus an immediate induction provides equality~(\ref{equ:eta}).

\begin{align}
  \lambda_{r,m}(\sigma w)&=\Lambda_{r,m,\sigma}(\lambda_{r,m}(w))\quad\quad &\forall \sigma\in\Sigma_r^*\quad\forall w\in\Sigma_r^\omega\label{equ:rho}\\
  \gamma_{r,m}(s,\sigma)&=\Gamma_{r,m,\sigma}(\frac{s}{1-r})\quad\quad &\forall\sigma\in(\Sigma_r^m)^*\quad\forall s\in S_r^m\label{equ:eta}
\end{align}

\medskip

We now reduce the computation of the closed convex hull $C$ of a set $X\subseteq \setR^m$ represented by an arithmetic automaton $A=(Q,\Sigma,T,Q_0,\F)$ in basis $r$ to data-flow analysis problems. We can assume w.l.o.g that $(Q,\Sigma,T)$ is a $m$-graph. As the language recognized by $A$ is included in $S_r^m\star(\Sigma_r^m)^*\star\Sigma_r^\omega$, the set of states can be partitioned into sets depending intuitively on the number of occurrences $|\sigma|_\star$ of the $\star$ symbol in a word $\sigma\in\Sigma^*$. More formally, we consider the set $Q_S$ of states reading \emph{signs}, the set $Q_I$ reading \emph{integers}, and the set $Q_D$ reading \emph{decimals} defined by:
\begin{align*}
  &Q_S=\{q\in Q \mid \exists (q_0,\sigma,F)\in Q_0\times \Sigma^*\times \F\quad|\sigma|_\star=0\quad\wedge\quad q_0\xrightarrow{\sigma}q\rightarrow F\}\\
  &Q_I=\{q\in Q \mid \exists (q_0,\sigma,F)\in Q_0\times \Sigma^*\times \F\quad|\sigma|_\star=1\quad\wedge\quad q_0\xrightarrow{\sigma}q\rightarrow F\}\\
  &Q_D=\{q\in Q \mid \exists (q_0,\sigma,F)\in Q_0\times \Sigma^*\times \F\quad|\sigma|_\star=2\quad\wedge\quad q_0\xrightarrow{\sigma}q\rightarrow F\}
\end{align*}
We also consider the $m$-graphs $G_S$, $G_I$ and $G_D$ obtained by restricting $G$ respectively to the states $Q_S$, $Q_I$ and $Q_D$ and formally defined by:
\begin{align*}
  &G_S=(Q_S,\Sigma_r,T_S) && \text{ with } T_S=T\cap (Q_s\times\Sigma_r\times Q_S)\\
  &G_I=(Q_I,\Sigma_r,T_I) && \text{ with } T_I=T\cap (Q_I\times\Sigma_r\times Q_I)\\
  &G_D=(Q_D,\Sigma_r,T_D) && \text{ with } T_D=T\cap (Q_D\times\Sigma_r\times Q_D)
\end{align*}

\begin{example}\label{ex:qsqi}
  $Q_S=\{-2,-1,0\}$, $Q_I=\{1,\ldots,9\}$ and $Q_D=\{a,b\}$ in Fig.~\ref{fig:ex}. 
\end{example}

The closed convex hull $C=\cloconv{X}$ is obtained from the valuations $C_I\in Q_I\rightarrow \C_m$ and $C_D\in Q_D\rightarrow\C_m$ defined by $C_I=\cloconv{X_I}$ and $C_D=\cloconv{X_D}$ where $X_I$ and $X_D$ are given by:
\begin{align*}
  X_I(q_I)&=\{\Gamma_{r,m,\sigma}(\frac{s}{1-r}) \mid s\in S_r^m\quad \sigma\in \Sigma_r^* \quad \exists q_0\in Q_0\quad q_0\xrightarrow{s\star \sigma}q_I\}\\
  X_D(q_D)&=\{\lambda_{r,m}(w) \mid w\in \Sigma_r^\omega \quad \exists F\in \F\quad q_D\xrightarrow{w}F\}\\
\end{align*}
In fact from the definition of arithmetic automata we get:
$$
C=\bigsqcup_{\stackrel{(q_I,q_D)\in Q_I\times Q_D}{(q_I,\star,q_D)\in T}}C_I(q_I)-C_D(q_D)
$$

\medskip

We now provide data-flow analysis problems whose $C_I$ and $C_D$ are solutions. Observe that $m$-graphs naturally denote control-flow graphs. Before associating semantics to $m$-graph transitions, we first show that $C_I$ and $C_D$ are some fix-point solutions. As $\cloconvop$ and $\Gamma_{r,m,a}$ are commutative, from the inclusion $\Gamma_{r,m,a}(X_I(q_1))\subseteq X_I(q_2)$ we deduce that $C_I$ satisfies the relation $\Gamma_{r,m,a}(C_I(q_1))\sqsubseteq C_I(q_2)$ for any transition $(q_2,a,q_2)\in T_I$. Symmetrically, as $\cloconvop$ and $\Lambda_{r,m,a}$ are commutative, from the inclusion $\Lambda_{r,m,a}(X_D(q_2))\subseteq X_D(q_1)$, we deduce that $\Lambda_{r,m,a}(C_D(q_2))\sqsubseteq C_D(q_1)$ for any transition $(q_1,a,q_2)\in T_D$. Intuitively $C_I$ and $C_D$ are two fix-point solutions of different systems. More formally, we associate two distinct semantics to a transition $t=(q_1,a,q_2)$ of a $m$-graph $G=(Q,\Sigma_r,T)$ by considering the monotonic functions $\Lambda_{G,m,t}$ and $\Gamma_{G,m,t}$ over the complete lattice ${(Q\rightarrow\C_m,\sqsubseteq)}$ defined for any $C\in Q\rightarrow\C_m$ and for any $q\in Q$ by the following equalities:
\begin{align*}
\Lambda_{G,m,t}(C)(q)=&
\begin{cases} 
\Lambda_{r,m,a}(C(q_2)) & \text{ if } q=q_1\\
C(q) & \text{ if } q\not=q_1\\
\end{cases}
\\
\Gamma_{G,m,t}(C)(q)=&
\begin{cases} 
\Gamma_{r,m,a}(C(q_1)) & \text{ if } q=q_2\\
C(q) & \text{ if } q\not=q_2\\
\end{cases}
\end{align*}
Observe that $C_D$ is a fix-point solution of the data-flow problem $\Lambda_{G_D,m,t}(C_D)\sqsubseteq C_D$ for any transition $t\in T_D$ and $C_I$ is a fix-point solution of the data-flow problem $\Gamma_{G_I,m,t}(C_I)\sqsubseteq C_I$ for any transition $t\in T_I$. In the next sections \ref{sec:CD} and \ref{sec:CI} we show that $C_D$ and $C_I$ can be characterized by these two data-flow analysis problems.

\subsection{Reduction for $C_D$}\label{sec:CD}
The computation of $C_D$ is reduced to a data-flow analysis problem for the $m$-graph $G_D$ equipped with the semantics $(\Lambda_{G_D,m,t})_{t\in T_D}$. 

\medskip

Given an infinite path $\theta$ labelled by $w$, we denote by $\lambda_{r,m}(\theta)$ the vector $\lambda_{r,m}(w)$. Given a $m$-graph $G$ labelled by $\Sigma_r$, we denote by $\Lambda_{G,m}$, the valuation $\cloconv{\lambda_{r,m}(\Theta_G)}$ (recall that $\Theta_G(q)$ denotes the set of infinite paths starting from $q$). This notation is motivated by the following Proposition \ref{prop:minimalC}.
\begin{proposition}\label{prop:minimalC}
  The valuation $\Lambda_{G,m}$ is the unique minimal valuation $C\in Q\rightarrow \C_m$ such that $\Lambda_{G,m,t}(C)\sqsubseteq C$ for any transition $t\in T$ and such that $C(q)\not=\emptyset$ for any state $q\in Q$ satisfying $\Theta_G(q)\not=\emptyset$.
\end{proposition}
The following Proposition \ref{prop:ttred1} provides the reduction.
\begin{proposition}\label{prop:ttred1}
  $C_D=\Lambda_{G_D,m}$
\end{proposition}
\begin{proof}
  We have previously proved that $\Lambda_{G_D,m,t}(C_D)\sqsubseteq C_D$ for any transition $t\in T_D$. Moreover, as $C_D(q_D)\not=\emptyset$ for any $q_D\in Q_D$, we deduce the relation $\Lambda_{G_D,m}\sqsubseteq C_D$ by minimality of $\Lambda_{G_D,m}$. For the other relation, just observe that $X_D\subseteq \lambda_{r,m}(\Theta_{G_D})$ and apply $\cloconvop$.
  \qed
\end{proof}

\subsection{Reduction for $C_I$}\label{sec:CI}
The computation of $C_I$ is reduced to data-flow analysis problems for the $m$-graphs $G_S$ and $G_I$ respectively equipped with the semantics $(\Gamma_{G_S,m,t})_{t\in T_S}$ and $(\Gamma_{G_I,m,t})_{t\in T_I}$. 

\medskip

Given a $m$-graph $G=(Q,\Sigma_r,T)$ and an \emph{initial valuation} $C_0\in Q\rightarrow\C_m$, it is well-known from Knaster-Tarski's theorem that there exists a unique minimal valuation $C\in Q\rightarrow\C_m$ such that $C_0\sqsubseteq C$ and $\Gamma_{G,m,t}(C)\sqsubseteq C$ for any $t\in T$. We denote by $\Gamma_{G,m}(C_0)$ this unique valuation.

\medskip

Symmetrically to the definitions of $C_I$ and $C_D$ we also consider the valuation $C_S\in Q_S\rightarrow\C_m$ defined by $C_S=\cloconv{X_S}$ where $X_S$ is given by: 
$$
X_S(q_S)=\{\Gamma_{r,m,s}(0,\ldots,0) \mid s\in S_r^* \quad \exists q_0\in Q_0\quad q_0\xrightarrow{s}q_S\}
$$

The reduction comes from the following Proposition \ref{prop:ttred2} where $C_{S,0}\in Q_S\rightarrow \C_m$ and $C_{I,0}\in Q_I\rightarrow \C_m$ are the following two initial valuations:
\begin{align*}
  C_{S,0}(q_S)&=
  \begin{cases}
    \emptyset & \text{ if } q_S\not\in Q_0\\
    \{(0,\ldots,0)\}     & \text{ if } q_S\in Q_0\\
  \end{cases} 
  \\
  C_{I,0}(q_I)&= \frac{1}{1-r}\bigsqcup_{\stackrel{q_S\in Q_S}{(q_S,\star,q_I)\in T}}C_S(q_S)
\end{align*}
\begin{proposition}\label{prop:ttred2}
  $C_S=\Gamma_{G_S,m}(C_{S,0})$ and $C_I=\Gamma_{G_I,m}(C_{I,0})$.
\end{proposition}
\begin{proof}
  First observe that $X_S\subseteq \Gamma_{G_S,m}(C_{S,0})$ and $X_I\subseteq \Gamma_{G_I,m}(C_{I,0})$. Thus $C_S\sqsubseteq \Gamma_{G_S,m}(C_{S,0})$ and $C_I\sqsubseteq \Gamma_{G_I,m}(C_{I,0})$ by applying $\cloconvop$. Finally, as $\Gamma_{r,m,a}$ and $\cloconvop$ are commutative, we deduce that $\Gamma_{G_S,m,t}(C_S)\sqsubseteq C_S$ for any $t\in T_S$ and $\Gamma_{G_I,m,t}(C_I)\sqsubseteq C_I$ for any $t\in T_I$. The minimality of $\Gamma_{G_S,m}(C_{S,0})$ and $\Gamma_{G_I,m}(C_{I,0})$ provide $\Gamma_{G_S,m}(C_{S,0})\sqsubseteq C_S$ and  $\Gamma_{G_I,m}(C_{I,0})\sqsubseteq C_I$.
  \qed
\end{proof}

\section{Infinite Paths Convex Hulls}\label{sec:infinite}
In this section $G=(Q,\Sigma_r,T)$ is a $m$-graph. We prove that $\Lambda_{G,m}(q)$ is equal to the convex hull of a finite set of rational vectors. Moreover, we provide an algorithm for computing the minimal sets $\Lambda_{G,m}^0(q)\subseteq \setQ^m$ for every $q\in Q$ such that $\Lambda_{G,m}=\conv{\Lambda_{G,m}^0}$ in exponential time in the worst case. 

\medskip

A \emph{fry-pan} $\theta$ in a graph $G$ is an infinite path $\theta=t_1\ldots t_{i}(t_{i+1}\ldots t_k)^\omega$ where $0\leq i<k$ and where $t_1=(q_0\rightarrow q_1),\ldots t_{k}=(q_{k-1}\rightarrow q_k)$ are transitions such that $q_k=q_i$. A fry-pan is said \emph{simple} if $q_0,\ldots,q_{k-1}$ are distinct states. The \emph{finite set of simple fry-pans} starting from $q$ is denoted by $\Theta_G^S(q)$. As expected, we are going to prove that $\Lambda_{G,m}=\conv{\lambda_{r,m}(\Theta_G^S)}$ and $\lambda_{r,m}(\Theta_G^S(q))\subseteq \setQ^m$.

\medskip

We first prove that $\lambda_{r,m}(\theta)$ is rational for any fry-pan $\theta$. Given $\sigma\in\Sigma_r^+$, the following Lemma \ref{lem:car} shows that $\lambda_{r,m}(\sigma^\omega)$ is the unique solution of the rational linear system $\Lambda_{r,m,\sigma}(x)=x$. In particular $\lambda_{r,m}(\sigma^\omega)$ is a rational vector. From equality (\ref{equ:rho}) given in page \pageref{equ:rho}, we deduce that the vector $\lambda_{r,m}(\theta)$ is rational for any fry-pan~$\theta$.
\begin{lemma}\label{lem:car}
  $\lambda_{r,m}(\sigma^\omega)$ is the unique fix-point of $\Lambda_{r,m,\sigma}$ for any $\sigma\in\Sigma_r^+$. 
\end{lemma}

\begin{figure}[htbp]
  \begin{center}  
    \begin{tabular}{@{}|c|@{}|c|@{}|c|@{}|c|@{}}
      \hline
      \begin{pspicture}(2.4,0.7)(-0.5,-0.7)
        \rput(-0.2,0){$q'$}
        \rput(2.2,0){$q$}
        \cnode(0,0){2pt}{qp}
        \cnode[linestyle=none](1,0){1pt}{qin}
        \cnode(2,0){2pt}{q}
        \ncarc[arcangle=-90]{->}{q}{qin}
        \ncline{->}{qin}{q}
        \ncline{->}{qp}{qin}
        \ncarc[linestyle=dotted,arcangle=90]{->}{q}{qp}
      \end{pspicture}
      &
      \begin{pspicture}(2.4,0.7)(-0.5,-0.7)
        \rput(-0.2,0){$q'$}
        \rput(2.2,0){$q$}
        \cnode(0,0){2pt}{qp}
        \cnode[linestyle=none](1,0){1pt}{qin}
        \cnode(2,0){2pt}{q}
        \ncarc[linestyle=dotted,arcangle=-90]{->}{q}{qin}
        \ncline[linestyle=dotted]{->}{qin}{q}
        \ncline[linestyle=dotted]{->}{qp}{qin}
        \ncarc[arcangle=90]{->}{q}{qp}
      \end{pspicture}
      &
      \begin{pspicture}(2.4,0.7)(-0.5,-0.7)
        \rput(-0.2,0){$q'$}
        \rput(2.2,0){$q$}
        \cnode(0,0){2pt}{qp}
        \cnode[linestyle=none](1,0){1pt}{qin}
        \cnode(2,0){2pt}{q}
        \ncarc[linestyle=dotted,arcangle=-90]{->}{q}{qin}
        \ncline{->}{qin}{q}
        \ncline{->}{qp}{qin}
        \ncarc[arcangle=90]{->}{q}{qp}
      \end{pspicture}
      &
      \begin{pspicture}(2.4,0.7)(-0.5,-0.7)
        \rput(-0.2,0){$q'$}
        \rput(2.2,0){$q$}
        \cnode(0,0){2pt}{qp}
        \cnode[linestyle=none](1,0){1pt}{qin}
        \cnode(2,0){2pt}{q}
        \ncarc[arcangle=-90]{->}{q}{qin}
        \ncline{->}{qin}{q}
        \ncline[linestyle=dotted]{->}{qp}{qin}
        \ncarc[linestyle=dotted,arcangle=90]{->}{q}{qp}
      \end{pspicture}\\
      \hline
      $\theta'$ & $t$ & $\pi$ & $\theta$\\
      \hline
      \begin{pspicture}(2.4,0.7)(-0.1,-0.7)
        \rput(0,0.3){$q'$}
        \rput(1,0.3){$q$}
        \cnode(0,0){2pt}{qp}
        \cnode[linestyle=none](2,0){1pt}{qin}
        \cnode(1,0){2pt}{q}
        \ncline[arcangle=-90]{->}{q}{qin}
        \ncline{->}{qp}{q}
        \nccircle[angle=-90]{->}{qin}{0.25}
        \ncarc[linestyle=dotted,arcangle=90]{->}{q}{qp}
      \end{pspicture}
      &
      \begin{pspicture}(2.4,0.7)(-0.1,-0.7)
        \rput(0,0.3){$q'$}
        \rput(1,0.3){$q$}
        \cnode(0,0){2pt}{qp}
        \cnode[linestyle=none](2,0){1pt}{qin}
        \cnode(1,0){2pt}{q}
        \ncline[linestyle=dotted,arcangle=-90]{->}{q}{qin}
        \ncline[linestyle=dotted]{->}{qp}{q}
        \nccircle[linestyle=dotted,angle=-90]{->}{qin}{0.25}
        \ncarc[arcangle=90]{->}{q}{qp}
      \end{pspicture}
      &
      \begin{pspicture}(2.4,0.7)(-0.1,-0.7)
        \rput(0,0.3){$q'$}
        \rput(1,0.3){$q$}
        \cnode(0,0){2pt}{qp}
        \cnode[linestyle=none](2,0){1pt}{qin}
        \cnode(1,0){2pt}{q}
        \ncline[linestyle=dotted,arcangle=-90]{->}{q}{qin}
        \ncline{->}{qp}{q}
        \nccircle[linestyle=dotted,angle=-90]{->}{qin}{0.25}
        \ncarc[arcangle=90]{->}{q}{qp}
      \end{pspicture}
      &
      \begin{pspicture}(2.4,0.7)(-0.1,-0.7)
        \rput(0,0.3){$q'$}
        \rput(1,0.3){$q$}
        \cnode(0,0){2pt}{qp}
        \cnode[linestyle=none](2,0){1pt}{qin}
        \cnode(1,0){2pt}{q}
        \ncline[arcangle=-90]{->}{q}{qin}
        \ncline[linestyle=dotted]{->}{qp}{q}
        \nccircle[angle=-90]{->}{qin}{0.25}
        \ncarc[linestyle=dotted,arcangle=90]{->}{q}{qp}
      \end{pspicture}
      \\
      \hline
    \end{tabular}
  \end{center}
  \caption{A graphical support for Proposition \ref{prop:simple} where $\theta'$ denotes a simple fry-pan starting from a state $q'$ and $t=(q,a,q')$ is a transition such that the fry-pan $t\theta'$ is not simple. That means the state $q$ is visited by $\theta'$. Note that $q$ is visited either once or infinitely often. These two situations are depicted respectively on the top line and the bottom line of the tabular.\label{fig:prop}}
\end{figure}
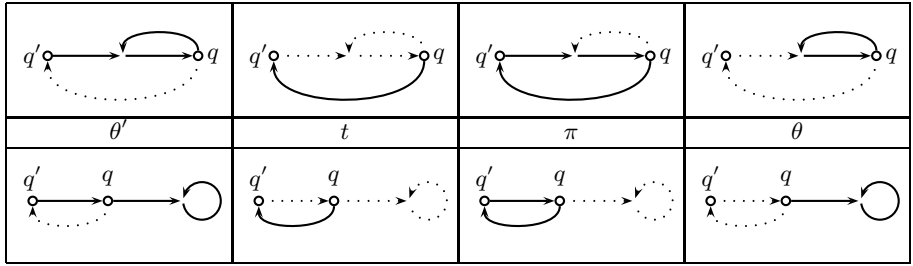

The following Proposition \ref{prop:simple} (see the graphical support given in Fig.~\ref{fig:prop}) is used in the sequel for effectively computing $\Lambda_{G,m}$ thanks to a fix-point iteration algorithm.
\begin{proposition}\label{prop:simple}
  Let $t=(q,a,q')$ be a transition and let $\theta'$ be a simple fry-pan starting from $q'$ such that the fry-pan $t\theta'$ is not simple. In this case there exists a minimal non-empty prefix $\pi$ of $t\theta'$ terminating in $q$. Moreover the fry-pan $\theta$ such that $t\theta'=\pi\theta$ and the fry-pan $\pi^\omega$ are simple and such that $\Lambda_{r,m,a}(\lambda_{r,m}(\theta'))\in\conv{\{\lambda_{r,m}(\theta),\lambda_{r,m}(\pi^\omega)\}}$.
\end{proposition}
\begin{proof}
  As $t\theta'$ is not simple whereas $\theta'$ is simple we deduce that there exists a decomposition of $t\theta'$ into $\pi\theta$ where $\pi$ is the minimal non-empty prefix of $t\theta'$ terminating in $q$. Let $\pi$ be the non empty path with the minimal length. Observe that $\pi$ is a simple cycle and thus $\pi^\omega$ is a simple fry-pan. Moreover, as $\theta$ is a suffix of the simple fry-pan $\theta'$, we also deduce that $\theta$ is a simple fry-pan. Observe that $\lambda_{r,m}(t\theta')=\lambda_{r,m}(\pi\theta)$. Moreover, as $\pi$ is a cycle in a $m$-graph we deduce that $m$ divides its length. Denoting by $\sigma$ the label of $\pi$, we deduce that $\sigma\in(\Sigma_r^m)^+$. Now, observe that $\Lambda_{r,m,\sigma}(x)=(1-r^{-\frac{|\sigma|}{m}})\lambda_{r,m}(\sigma^\omega)+r^{-\frac{|\sigma|}{m}}x$ for any $x\in\setR^m$. We deduce that $\Lambda_{r,m,a}(\lambda_{r,m}(\theta'))=(1-r^{-\frac{|\sigma|}{m}})\lambda_{r,m}(\pi^\omega)+r^{-\frac{|\sigma|}{m}}\lambda_{r,m}(\theta)$. Thus $\Lambda_{r,m,a}(\lambda_{r,m}(\theta'))\in\conv{\{\lambda_{r,m}(\theta),\lambda_{r,m}(\pi^\omega)\}}$.
  \qed
\end{proof}

From the previous Proposition \ref{prop:simple} we deduce the following Proposition \ref{prop:cont}.
\begin{proposition}\label{prop:cont}
  We have $\Lambda_{G,m}=\conv{\lambda_{r,m}(\Theta_G^S)}$.
\end{proposition}

We deduce that there exists a minimal finite set $\Lambda_{G,m}^0(q)\subseteq\setQ^m$ such that $\Lambda_{G,m}=\conv{\Lambda_{G,m}^0}$. Note that an exhaustive computation of the whole set $\Theta_G^S(q)$ provides the set $\Lambda_{G,m}^0(q)$ by removing vectors that are convex combination of others. The efficiency of such an algorithm can be greatly improved by computing inductively subsets $\Theta(q)\subseteq \Theta_G^S(q)$ and get rid of any fry-pan $\theta\in \Theta(q)$ as soon as it becomes a convex combination of other fry-pans in $\Theta(q)\moins\{\theta\}$. The algorithm \textsf{Cycle} is based on this idea.
\begin{corollary}\label{cor:infinite}
  The algorithm \textsf{Cycle}($G$,$m$) terminates by iterating the main while loop at most $|T|^{|Q|}$ times and it returns $\Lambda_{G,m}^0$.
\end{corollary}

\begin{lstlisting}[emph=Cycle]
Cycle($G=(Q,\Sigma_r,T)$ be a $m$-graph, $m\in\setN\moins\{0\}$)
   for each state $q\in Q$
      if $\Theta_G^S(q)\not=\emptyset$
         let $\theta\in\Theta_G^S(q)$
	 let $\Theta(q)\leftarrow\{\theta\}$
      else
         let $\Theta(q)\leftarrow\emptyset$
   while there exists $t=(q,a,q')\in T$ and $\theta'\in\Theta(q')$
      such that $\Lambda_{r,m,a}(\lambda_{r,m}(\theta'))\not\in \conv{\lambda_{r,m}(\Theta(q))}$
      if $t\theta'$ is simple
         let $\Theta(q)\leftarrow \Theta(q)\cup\{t\theta'\}$
      else
         let $\pi$ be the minimal strict prefix of $t\theta'$ terminating in $q$ 
         let $\theta$ be such that $t\theta'=\pi\theta$
	 let $\Theta(q)\leftarrow\Theta(q)\cup\{\theta,\pi^\omega\}$
      while there exists $\theta_0\in \Theta(q)$ 
         such that $\conv{\lambda_{r,m}(\Theta(q))}=\conv{\lambda_{r,m}(\Theta(q)\moins\{\theta_0\})}$
         let $\Theta(q)\leftarrow \Theta(q)\moins \{\theta_0\}$
   return $\lambda_{r,m}(\Theta)$                                 //$\Lambda_{G,m}^0$
\end{lstlisting}

\section{Fix-point Computation}\label{sec:computation}
In this section we prove that the minimal post-fix-point $\Gamma_{G,m}(C_0)$ is effectively rational polyhedral for any $m$-graph $G=(Q,\Sigma_r,T)$ and for any rational polyhedral initial valuation $C_0\in Q\rightarrow\C_m$. We deduce that the closed convex hull of sets symbolically represented by arithmetic automata are effectively rational polyhedral.

\begin{example}\label{ex:kleene}
  Let $m=1$ and $G=(\{q\},\Sigma_r,\{t\})$ where $t=(q,r-1,q)$ and $C_0(q)=\{0\}$. Observe that the sequence $(C_{i})_{i\in\setN}$ where $C_{i+1}=C_i\sqcup \Gamma_{G,m,t}(C_i)$ satisfies $C_{i}(q)=\{x\in\setR \mid 0\leq x\leq r^i-1\}$.
\end{example}

\medskip

Recall that a Kleene iteration algorithm applied on the computation of $\Gamma_{G,m}(C_0)$ consists in computing the beginning of the sequence $(C_i)_{i\in\setN}$ defined by the induction $C_{i+1}=C_i\bigsqcup_{t\in T}\Gamma_{G,m,t}(C_i)$ until an integer $i$ such that $C_{i+1}=C_i$ is discovered. Then the algorithm terminates and it returns $C_i$. In fact, in this case we have $C_i=\Gamma_{G,m}(C_0)$. However, as proved by the previous Example \ref{ex:kleene} the Kleene iteration does not terminate in general. Nevertheless we are going to compute $\Gamma_{G,m}(C_0)$ by a Kleene iteration such that each $C_i$ is \emph{safely} enlarged into a $C_i'$ satisfying $C_i\sqsubseteq C_i'\sqsubseteq \Gamma_{G,m}(C_0)$. This enlargement follows the acceleration framework introduced in \cite{DBLP:conf/fsttcs/LerouxS07,DBLP:conf/sas/LerouxS07} that roughly consists to compute the precise effect of iterating some cycles. This framework motivate the introduction of the monotonic function $\Gamma^W_{G,m}$ defined over the complete lattice $(Q\rightarrow\C_m,\sqsubseteq)$ for any $C\in Q\rightarrow\C_m$ and for any $q\in Q$ by the following equality:
$$\Gamma_{G,m}^W(C)(q)=\bigsqcup_{q\xrightarrow{\sigma}q}\Gamma_{r,m,\sigma}(C(q))$$

\medskip

The following Proposition \ref{prop:unlargeW} shows that $\Gamma_{G,m}^W(C)$ is effectively computable from $C$ and the function $\Lambda_{G,m}$ introduced in section \ref{sec:reduction}. In this proposition, $G_q$ denotes the graph $G$ reduced to the strongly connected components of $q$.
\begin{proposition}\label{prop:unlargeW}
  For any $C\in Q\rightarrow\C_m$, and for any $q\in Q$, we have:
  $$
  \Gamma_{G,m}^W(C)(q)=C(q)+\setR_+(C(q)-\Lambda_{G_q,m}(q))
  $$
\end{proposition}

We now prove that the enlargement is sufficient to enforce the convergence of a Kleene iteration.
\begin{proposition}\label{prop:unlarge}
  Let $C_0\sqsubseteq C_0'\sqsubseteq C_1\sqsubseteq C_1'\sqsubseteq \ldots$ be the sequence defined by the induction $C_{i+1}=C_i'\bigsqcup_{t\in T}\Gamma_{G,m,t}(C_i')$ and $C'_i=\Gamma_{G,m}^W(C_i)$. There exists $i< |Q|$ satisfying $C_{i+1}=C_i$. Moreover, for such an integer $i$ we have $C_i=\Gamma_{G,m}(C_0)$.
\end{proposition}
\begin{proof}
  Observe that $C_i\sqsubseteq C'_i\sqsubseteq \Gamma_{G,m}(C_0)$ for any $i\in\setN$. Thus, if there exists $i\in\setN$ such that $C_{i+1}=C_i$ we deduce that $C_i=\Gamma_{G,m}(C_0)$. Finally, in order to get the equality $C_{|Q|}=C_{|Q|-1}$, just observe by induction over $i$ that we have following equality for any $q_2\in Q$:
  $$
  C_i'(q_2)=
  \bigsqcup_{
    \stackrel{q_0\xrightarrow{\sigma_1}q_1\xrightarrow{\sigma}q_1\xrightarrow{\sigma_2}q_2}{|\sigma_1|+|\sigma_2|\leq i}%
  }
  \Gamma_{G,m,\sigma_1\sigma\sigma_2}(C_0(q_1))
  $$
  \qed
\end{proof}

\begin{example}\label{ex:lastb}
  Let us consider the $2$-graph $G_I$ obtained from the $2$-graph depicted in the center of Fig.~\ref{fig:ex} and restricted to the set of states $Q_I=\{1,\ldots,9\}$. Let us also consider the function $C_{I,0}\in Q_I\rightarrow\C_2$ defined by $C_{I,0}(1)=\{(0,0)\}$ and $C_{I,0}(q)=\emptyset$ for $q\in\{2,\ldots,9\}$. Computing inductively the sequence $C_0\sqsubseteq C_0'\sqsubseteq C_1\sqsubseteq C_1'\sqsubseteq \ldots$ defined in Proposition \ref{prop:unlarge} from $C_0=C_{I,0}$ shows that $C_6=C_5$\ifthenelsepaperversion{proc}{}{ (see section \ref{sec:execution} in appendix)}. Moreover, this computation provides the value of $C_I=\Gamma_{G_I,2}(C_{I,0})$ (see Table~\ref{tab:val}).
\end{example}

\begin{table}[Htbp]
  \centering
  $\begin{array}{|c|c|c|}
    \hline
    q   & C_{I,0}(q) & \Gamma_{G_I,2}(C_{I,0})(q) \\
    \hline
    \hline
    1   & \{(0,0)\} & \setR_+(1,3) \\
    2   & \emptyset & (1,1)+\setR_+(3,2)\\
    3   & \emptyset & \setR_+(3,2)\\
    4   & \emptyset & (1,0)+\setR_+(3,2) \\
    5   & \emptyset & (0,1)+\setR_+(1,3) \\
    6   & \emptyset & (2,1)+\setR_+(3,2) \\
    7   & \emptyset & (0,2)+\setR_+(1,3) \\
    8   & \emptyset & \conv{\{(1,0),(1,2)\}}+\setR_+(1,0)+\setR_+(1,3) \\
    9   & \emptyset & (0,1)+\setR_+(0,1)+\setR_+(3,2) \\
    \hline
  \end{array}$
  \caption{The values of $C_{I,0}$ and $C_I=\Gamma_{G_I,2}(C_{I,0})$.\label{tab:val}}
\end{table}

\begin{lstlisting}[emph=Cycle,emph=FixPoint]
FixPoint($G=(Q,\Sigma_r,T)$ a $m$-graph, $m\in\setN\moins\{0\}$, $C_0\in Q\rightarrow\C_m$)
   let $C\leftarrow C_0$
   while there exists $t\in T$ such that $\Gamma_{G,m,t}(C)\not\sqsubseteq C$ 
      $C\leftarrow \Gamma_{G,m}^W(C)$
      let $C\leftarrow C\sqcup\bigsqcup_{t\in T}\Gamma_{G,m,t}(C)$
   return $C$
\end{lstlisting}

\begin{corollary}\label{cor:finite}
  The algorithm \textsf{FixPoint}($G$,$m$,$C_0$) terminates by iterating the main while loop at most $|Q|-1$ times. Moreover, the algorithm returns $\Gamma_{G,m}(C_0)$.
\end{corollary}

From Propositions \ref{prop:ttred1} and \ref{prop:ttred2} and corollaries \ref{cor:infinite} and \ref{cor:finite} we get:
\begin{theorem}\label{thm:main}
  The closed convex hull of sets symbolically represented by arithmetic automata are rational polyhedral and computable in exponential time.
\end{theorem}

\begin{example}
  We follow notations introduced in Examples \ref{ex:fin1}, \ref{ex:qsqi} and \ref{ex:lastb}. Observe that $C_I(8)-C_D(a)=\conv{\{(1,0),(1,2)\}}+\setR_+(1,0)+\setR_+(1,3)$ is exactly the closed convex hull of $X=\{x\in\setN^2 \mid 3x[1]>x[2]\}$.
\end{example}

\section{Conclusion}
We have proved that the closed convex hull of sets symbolically represented by arithmetic automata are rational polyhedral. Our approach is based on acceleration in convex data-flow analysis. It provides a simple algorithm for computing this set. Compare to \cite{DBLP:conf/lics/Latour04} (1) our algorithm has the same worst case exponential time complexity, (2) it is not limited to sets of the form  $\setZ^m\cap C$ where $C$ is a rational polyhedral convex set, (3) it can be applied to any set definable in $\fo{\setR,\setZ,+,\leq,X_r}$, (4) it can be easily implemented, and (5) it is not restricted to the most significant digit first decomposition. This last advantage directly comes from the class of arithmetic automata we consider. In fact, since the arithmetic automata can be non deterministic, our algorithm can be applied to least significant digit first arithmetic automata just by flipping the direction of the transitions. Finally, from a practical point of view, as the arithmetic automata representing sets in the restricted logic $\fo{\setR,\setZ,+,\leq}$ (where $X_r$ is discarded) have a very particular structure, we are confident that the exponential time complexity algorithm can be applied on automata with many states like the one presented in \cite{DBLP:conf/lics/Latour04}. The algorithm will be implemented in \textsc{TaPAS} \cite{LP-08} (The Talence Presburger Arithmetic Suite) as soon as possible.

\bibliographystyle{alpha}
\bibliography{thisbiblio}

\ifthenelsepaperversion{proc}{}{
\clearpage
\appendix
\appendix
\newcommand{\APP}[3]{\noindent\textbf{#1 #2.} \emph{#3}}

\newpage
\gdef\thesection{A}
\section{Proof of Proposition \ref{prop:minimalC}}

\APP{Proposition}{\ref{prop:minimalC}}{%
The valuation $\Lambda_{G,m}$ is the unique minimal valuation $C\in Q\rightarrow \C_m$ such that $\Lambda_{G,m,t}(C)\sqsubseteq C$ for any transition $t\in T$ and such that $C(q)\not=\emptyset$ for any state $q\in Q$ satisfying $\Theta_G(q)\not=\emptyset$.
}
\begin{proof}
  Let us first prove that $C=\cloconv{\lambda_{r,m}(\Theta_G)}$ is a valuation in $Q\rightarrow \C_m$ such that $\Lambda_{G,m,t}(C)\sqsubseteq C$ for any transition $t\in T$. We have the inclusion $\Lambda_{r,m,a}(\lambda_{r,m}(\Theta_G(q_2)))\subseteq \lambda_{r,m}(\Theta_G(q_1))$ for any transition $(q_1,a,q_2)\in T$. As $\cloconvop$ and $\Lambda_{r,m,a}$ are commutative, the valuation $C=\cloconv{\lambda_{r,m}(\Theta_G)}$ satisfies $\Lambda_{G,m,t}(C)\sqsubseteq C$ for any transition $t\in T$.

  Now, let us consider a valuation $C\in Q\rightarrow \C_m$ such that $\Lambda_{G,m,t}(C)\sqsubseteq C$ for any transition $t\in T$ and such that $C(q)\not=\emptyset$ for any state $q\in Q$ satisfying $\Theta_G(q)\not=\emptyset$. Let us prove that $\cloconv{\lambda_{r,m}(\Theta_G)}\sqsubseteq C$. As $\Lambda_{G,m,t}(C)\sqsubseteq C$ for any transition $t\in T$ an immediate induction shows that $\Lambda_{r,m,\sigma}(C(q))\sqsubseteq C(q')$ for any finite path $\pi=(q\xrightarrow{\sigma}q')$. Let us consider an infinite path $\theta=(q\xrightarrow{w}F)$. As $F$ is non empty, there exists a state $q'\in F$. Recall that $F$ is the set of states visited infinitely often by the path $\theta$. We deduce that there exists a cycle on $q'$ and in particular $\Theta_G(q')\not=\emptyset$. This condition implies $C(q')\not=\emptyset$. Thus there exists $x'\in C(q')$. Moreover, as $q'$ is visited infinitely often by $\theta$, there exists a strictly increasing sequence $0\leq i_0<i_1<\cdots$ of integers such that $q\xrightarrow{w(1)\ldots w(i_j)}q'$. This path shows that the vector $x_j=\Lambda_{r,m,w(1)\ldots w(i_j)}(x')$ is in $C(q)$. As $\lim_{j\rightarrow+\infty}x_j=\lambda_{r,m}(w)$ and $C(q)$ is closed we deduce that $\lambda_{r,m}(w)\in C(q)$. We have proved that $\lambda_{r,m}(\Theta_G)\subseteq C$. Therefore $\cloconv{\lambda_{r,m}(\Theta_G)}\sqsubseteq C$.
  \qed
\end{proof}


\newpage
\gdef\thesection{B}
\section{Proof of Lemma \ref{lem:car}}

\APP{Lemma}{\ref{lem:car}}{%
  $\lambda_{r,m}(\sigma^\omega)$ is the unique fix-point of $\Lambda_{r,m,\sigma}$ for any $\sigma\in\Sigma_r^+$. 
}
\begin{proof}
  As $\sigma\sigma^\omega$ and $\sigma^\omega$ are equal, equality (\ref{equ:rho}) page \pageref{equ:rho} shows that $\lambda_{r,m}(\sigma^\omega)$ is a fix-point of $\Lambda_{r,m,\sigma}$. Moreover as the uniform form of the linear function $\Lambda_{r,m,a}$ is equal to $\Lambda_{r,m,0}$ we deduce that the uniform form of $\Lambda_{r,m,\sigma}^m$ is equal to $\Lambda_{r,m,0}^{m|\sigma|}$. Since $\Lambda_{r,m,0}^m(x)=r^{-1}x$ we have proved that the uniform form of $\Lambda_{r,m,\sigma}^m$ is $x\rightarrow r^{-|\sigma|}x$ for any $x\in\setR^m$. Moreover, as $\lambda_{r,m}(\sigma^\omega)$ is a fix-point of $\Lambda_{r,m,\sigma}^m$ we deduce that $\Lambda_{r,m,\sigma}^m(x)=\lambda_{r,m}(\sigma^\omega)+r^{-|\sigma|}(x-\lambda_{r,m}(\sigma^\omega))$ for any $x\in\setR^m$. In particular, if $x$ is a fix-point of $\Lambda_{r,m,\sigma}$, we get $x=\lambda_{r,m}(\sigma^\omega)+r^{-|\sigma|}(x-\lambda_{r,m}(\sigma^\omega))$. As $r^{-|\sigma|}\not=1$ we obtain $x=\lambda_{r,m}(\sigma^\omega)$.
  \qed
\end{proof}

\newpage
\gdef\thesection{C}
\section{Proof of Proposition \ref{prop:cont}}

\APP{Proposition}{\ref{prop:cont}}{%
  We have $\Lambda_{G,m}=\conv{\lambda_{r,m}(\Theta_G^S)}$.
}
\begin{proof}
  From $\Theta_G^S(q)\subseteq\Theta_G(q)$ we deduce the inclusion $\conv{\lambda_{r,m}(\Theta_G^S)}\subseteq \Lambda_{G,m}$. Let us prove the other inclusion. Observe that $\Theta_G^S(q)$ is a finite set and in particular $\conv{\Theta_G^S(q)}$ is a closed convex set for any $q\in Q$. Let us consider the function $C\in Q\rightarrow \C_m$ defined by $C=\conv{\lambda_{r,m}(\Theta_G^S)}$. From Proposition \ref{prop:simple}, we deduce that $\Lambda_{G,m,t}(C)\sqsubseteq C$ for any transition $t\in T$. Note also that $\C(q)\not=\emptyset$ for any state $q\in Q$ such that $\Theta_G(q)\not=\emptyset$. By minimality of $\Lambda_{G,m}$ we get the other inclusion $\Lambda_{G,m}\sqsubseteq C$.
  \qed
\end{proof}

\newpage
\gdef\thesection{D}
\section{An Additional Example For Section \ref{sec:infinite}}

\begin{table}[htbp]
  \begin{center}
 
    $\begin{array}{|c|c|c|}
    \hline
    q 
       & \Theta_{G_1}^S(q)
       & -\Lambda_{G_1,2}^0(q)
    \\
    \hline
    1  
       & (00)^\omega,(0111)^\omega, 01(0010)^\omega, 0100(11)^\omega
       & \{(0,0),(\frac{1}{3},1))\}
    \\
    2  
       & 1(00)^\omega,(1011)^\omega,101(0010)^\omega,10100(11)^\omega
       & \{(\frac{1}{2},0),(\frac{7}{8},\frac{1}{4}))\}
    \\
    3  
       & (1110)^\omega,100(11)^\omega,1(0010)^\omega
       & \{(1,\frac{2}{3}),(\frac{1}{2},\frac{1}{3})\}
    \\
    4  
       & 0(11)^\omega, (0100)^\omega, 01011(00)^\omega, 010(1101)^\omega
       & \{(\frac{1}{2},1),\{(0,\frac{2}{3})\}
    \\
    5  
       & 11(00)^\omega,(1101)^\omega,(0010)^\omega,00(11)^\omega
       & \{(\frac{2}{3},1),(\frac{1}{3},0)\}
    \\
    6  
       & (11)^\omega, (0001)^\omega, 0(1101)^\omega, 011(00)^\omega
       & \{(1,1),(0,\frac{1}{3})\}
    \\
    7  
       & (11)^\omega, (1000)^\omega, 10(1101)^\omega, 1011(00)^\omega
       & \{(\frac{2}{3},0),(1,1)\}
    \\
    \hline
  \end{array}$
  \end{center}
  \caption{Some values computed\label{tab:ex}}
\end{table}

\begin{figure}        
\setlength{\unitlength}{10pt}
\pssetlength{\psunit}{10pt}
\pssetlength{\psxunit}{10pt}
\pssetlength{\psyunit}{10pt}
\begin{picture}(27,17)(-14,-9)%
\rput(0,0){\includegraphics{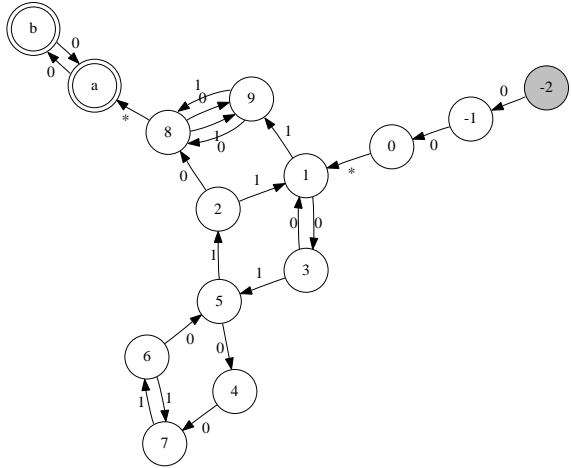}}
\end{picture}
\caption{An arithmetic automaton in basis $2$ and in dimension $2$.\label{fig:auto}}
\end{figure}

\begin{example}\label{ex:theta}
  Let us consider the $2$-graph $G$ labelled by $\Sigma_2$ and depicted in Fig.~\ref{fig:auto}. We denote by $G_1$ the graph $G$ restricted to the strongly connected component $\{1,\ldots,7\}$. By enumerating all the possible simple fry-pans $\Theta_{G_1}^S(q)$ starting from a state $q$, observe that we get the values given in the Table~\ref{tab:ex}. This table only provides the labels of the fry-pans in order to simplify the presentation. However, the fry-pans can be recovered from their labels since the graph is deterministic.
\end{example}

\newpage
\gdef\thesection{E}
\section{Proof of Corollary \ref{cor:infinite}}

\APP{Corollary}{\ref{cor:infinite}}{%
  The algorithm \textsf{Cycle}($G$,$m$) terminates by iterating the main while loop at most $|T|^{|Q|}$ times and it returns $\Lambda_{G,m}^0$.
}
\begin{proof}
  Observe that $\Theta(q)\subseteq \Theta_G^S(q)$ for any state $q$ at any step of the algorithm. Moreover, each time the while loop is executed, the set $C(q)=\conv{\lambda_{r,m}(\Theta(q))}$ strictly increases. Thus, the set $\{\theta\in \Theta_G^S(q) \mid \lambda_{r,m}(\theta)\in C(q)\}$ strictly increase each time the while loop is executed. Observe that a simple fry-pan $\theta$ is uniquelly determined from its $|Q|$ first transitions. Thus $\sum_{q\in Q}|\Theta_G^S(q)|\leq |T|^{|Q|}$. We deduce that the algorithm terminates after executing at most $|T|^{|Q|}$ times the while loop. Finally, let us prove that when the algorithm terminates it returns $\Lambda_{G,m}^0$. It is sufficient to show that $C=\Lambda_{G,m}$ when it terminates. Note that the while loop condition is no longer valid. Thus $\Lambda_{G,m,t}(C)\sqsubseteq C$ for any transition $t\in T$. As $C(q)\not=\emptyset$ for any state $q\in Q$ such that $\Theta_G(q)\not=\emptyset$, by minimality of $\Lambda_{G,m}$ we deduce that $\Lambda_{G,m}\sqsubseteq C$. Thus $C=\Lambda_{G,m}$ when the algorithm terminates. We deduce that the algorithm returns $\Lambda_{G,m}^0$.
  \qed
\end{proof}

\newpage
\gdef\thesection{F}
\section{Proof of Proposition \ref{prop:unlargeW}}

We first prove the following two technical lemmas.
\begin{lemma}\label{lem:tech}
  For any $\sigma\in(\Sigma_r^m)^+$ and for any $x\in\setR^m$ we have:
  $$
  \cloconv{\{\Gamma_{r,m,\sigma^i}(x) \mid i\in\setN\}}=x+\setR_+(x-\lambda_{r,m}(\sigma^\omega))
$$
\end{lemma}
\begin{proof}
  As $\lambda_{r,m}(\sigma^\omega)$ is a fix-point of the linear function $\Gamma_{r,m,\sigma^i}$ and as the uniform form of the linear function $\Gamma_{r,m,\sigma^i}$ is $\Gamma_{r,m,0}^{i|\sigma|}$, we deduce that $\Gamma_{r,m,\sigma^i}(x)=x+(r^{i\frac{|\sigma|}{m}}-1)(x-\lambda_{r,m}(\sigma^\omega))$ for any $i\in\setN$. As $\cloconv{\{r^{i\frac{|\sigma|}{m}}-1 \mid i\in\setN\}}=\setR_+$ we deduce the lemma.
  \qed
\end{proof}

\begin{lemma}\label{lem:xixi}
  For any strongly connected $m$-graph $G=(Q,\Sigma_r,T)$ and for any state $q\in Q$, we have :
  $$\Lambda_{G,m}(q)=\cloconv{\{\lambda_{r,m}(\sigma^\omega) \mid q\xrightarrow{\sigma\in(\Sigma_r^m)^+}q\}}$$
\end{lemma}
\begin{proof}
  Let $C(q)=\cloconv{\{\lambda_{r,m}(\sigma^\omega) \mid q\xrightarrow{\sigma\in(\Sigma_r^m)^+}q\}}$ be defined for any $q\in Q$. Note that for any cycle $\pi=(q\xrightarrow{\sigma\in(\Sigma_r^m)^+}q)$ we have $\pi^\omega\in \Theta_{G_q}(q)$. In particular $\lambda_{r,m}(\sigma^\omega)\in \Lambda_{G,m}(q)$. We deduce the inclusion $C(q)\sqsubseteq \Lambda_{G,m}(q)$. For the other inclusion, let us consider an infinite path $q\xrightarrow{w}F$ and let $q'\in F$. Since $F$ is the set of states visited infinitely often, there exists a strictly increasing sequence of integers $0< i_0<i_1<\cdots$ such that $q\xrightarrow{w(1)\ldots w(i_j)}q'$ for any integer $j\geq 0$. As $G$ is strongly connected, there exists a path $q'\xrightarrow{\sigma}q$. The cycle $q\xrightarrow{w(1)\ldots w(i_j)\sigma}q$ shows that the vector $x_j=\lambda_{r,m}((w(1)\ldots w(i_j)\sigma)^\omega)$ is in $C(q)$. As $\lim_{j\rightarrow+\infty}x_j=\lambda_{r,m}(w)$ and $C(q)$ is closed we have proved that $\lambda_{r,m}(w)\in C(q)$. Thus $\Lambda_{G,m}(q)\sqsubseteq C(q)$.
  \qed
\end{proof}

\APP{Proposition}{\ref{prop:unlargeW}}{%
  For any $C\in Q\rightarrow\C_m$, and for any $q\in Q$, we have :
  $$
  \Gamma_{G,m}^W(C)(q)=C(q)+\setR_+(C(q)-\Lambda_{G_q,m}(q))
  $$
}
\begin{proof}
  Note that if there does not exist a $q\xrightarrow{\sigma}q$ then $\Lambda_{G_q,m}(q)=\emptyset$ and the previous equality is immediate. Otherwise, from Lemmas \ref{lem:tech} and \ref{lem:xixi} we get the following equalities:
  \begin{align*}
    \Gamma_{G,m}^W(C)(q)
    &=\bigsqcup_{q\xrightarrow{\sigma}q}\Gamma_{r,m,\sigma}(C(q))\\
    &=\bigsqcup_{x\in C(q)}\bigsqcup_{q\xrightarrow{\sigma\in(\Sigma_r^m)^+}q}\cloconv{\{\Gamma_{r,m,\sigma^i}(x) \mid i\in\setN\}}\\
    &=\bigsqcup_{x\in C(q)}\bigsqcup_{q\xrightarrow{\sigma\in(\Sigma_r^m)^+}q}x+\setR_+(x-\lambda_{r,m}(\sigma^\omega))\\
    &=\bigsqcup_{x\in C(q)}x+\setR_+(x-\Lambda_{G_q,m}(q))\\
  \end{align*}
  In particular we deduce that $\Gamma_{G,m}^W(C)(q)\sqsubseteq C(q)+\setR_+(C(q)-\Lambda_{G_q.m}(q))$. Conversely, let us consider $x\in C(q)+\setR_+(C(q)-\Lambda_{G_q,m}(q))$. The vector $x$ can be decomposed into $x=c_1+h(c_2-z)$ where $c_1,c_2\in C(q)$, $z\in \Lambda_{G_q,m}(q)$ and $h\in\setR_+$. Let us denote by $c=\frac{1}{1+h}(c_1+hc_2)$. As $C(q)$ is convex we deduce that $c\in C(q)$. From $x=c+h(c-z)$ we deduce that $x\in \Gamma_{G,m}^W(C)(q)$.
  \qed
\end{proof}

\newpage
\gdef\thesection{G}
\section{An Execution of Algorithm \textsf{FixPoint}}\label{sec:execution}
\begin{sideways}
  \centering
  $\begin{array}{|c||c|c|c|c|c|c|c|c|c|}
    \hline
    & 1 & 2 & 3 & 4 & 5 & 6 & 7 & 8 & 9\\
    \hline
    \hline
    C_0 
        & \{(0,0)\} 
        & \emptyset 
        & \emptyset 
        & \emptyset 
        & \emptyset 
        & \emptyset 
        & \emptyset 
        & \emptyset 
        & \emptyset 
    \\
    \hline
    C_0' 
        & L 
        & \emptyset 
        & \emptyset 
        & \emptyset 
        & \emptyset 
        & \emptyset 
        & \emptyset 
        & \emptyset 
        & \emptyset 
    \\
    \hline
    C_1 
        & L 
        & \emptyset 
        & L' 
        & \emptyset 
        & \emptyset 
        & \emptyset 
        & \emptyset 
        & \emptyset 
        &  (0,1)+\setR_+(0,1)
    \\
    \hline
    C_1' 
        & L  
        & \emptyset 
        & L' 
        & \emptyset 
        & \emptyset 
        & \emptyset 
        & \emptyset 
        & \emptyset 
        &  (0,1)+D'
    \\
    \hline
    C_2 
        & L  
        & \emptyset 
        & L' 
        & \emptyset 
        & (0,1)+L 
        & \emptyset 
        & \emptyset 
        &  (1,1)+D
        &  (0,1)+D'
    \\
    \hline
    C_2' 
        & L 
        & \emptyset 
        & L' 
        & \emptyset 
        & (0,1)+L 
        & \emptyset 
        & \emptyset 
        & (1,1)+D 
        & (0,1)+D'
    \\
    \hline
    C_3 
        & L 
        & (1,1)+L' 
        & L' 
        & (1,0)+L' 
        & (0,1)+L 
        & \emptyset 
        & \emptyset 
        & \conv{\{(1,0),(1,1)\}}+D 
        & (0,1)+D' 
    \\
    \hline
    C_3' 
        & L 
        & (1,1)+L' 
        & L' 
        & (1,0)+L'  
        & (0,1)+L 
        & \emptyset 
        & \emptyset 
        & \conv{\{(1,0),(1,1)\}}+D 
        & (0,1)+D' 

    \\
    \hline
    C_4 
        &L  
        & (1,1)+L' 
        & L' 
        & (1,0)+L' 
        & (0,1)+L 
        & \emptyset 
        & (0,2)+L 
        & \conv{\{(1,0),(1,2)\}}+D 
        & (0,1)+D' 
    \\
    \hline
    C_4' 
        &L  
        & (1,1)+L' 
        & L' 
        & (1,0)+L' 
        & (0,1)+L 
        & \emptyset 
        & (0,2)+L 
        & \conv{\{(1,0),(1,2)\}}+D 
        & (0,1)+D' 
    \\
    \hline
    C_5 
        & L  
        & (1,1)+L' 
        & L' 
        & (1,0)+L' 
        & (0,1)+L 
        & (2,1)+L' 
        & (0,2)+L 
        & \conv{\{(1,0),(1,2)\}}+D 
        & (0,1)+D' 
    \\
    \hline
  \end{array}$
\end{sideways}\\
  Where $L=\setR_+(1,3)$, $L'=\setR_+(3,2)$, $D=\setR_+(1,0)+L$ and $D'=\setR_+(0,1)+L'$.

} 

\end{document}